\documentclass[12pt]{svmult}
\usepackage{amsfonts} 
\usepackage{amsmath}
\usepackage{bm}

\usepackage{graphicx}        

\usepackage[bottom]{footmisc}

\makeindex             

\def \Im   {\text {\rm Im}}

\def \Tr   {\text {\rm Tr}}
\def \Re   {\text {\rm Re}}

\begin{document}
\title*{Quantum Stochastic Calculus and Quantum Gaussian Processes}
\author{K. R. Parthasarathy}
\institute{Indian Statistical Institute, Delhi 
Centre, 7, S. J. S. Sansanwal Marg, New Delhi 110016, 
\texttt{krp@isid.ac.in}}
%
%
\maketitle
\vskip0.1in

\abstract{In this lecture we present a brief outline of boson Fock space 
stochastic calculus based on the creation, conservation and annihilation 
operators of free field theory, as given in the 1984 paper of Hudson 
and 
Parthasarathy \cite{9}. We show how a part of this architecture yields 
Gaussian fields stationary under a group action. Then we introduce the notion 
of semigroups of quasifree completely positive maps on the algebra of all 
bounded operators in the boson Fock space $\Gamma (\mathbb{C}^n)$ over 
$\mathbb{C}^n.$  These semigroups are not strongly continuous but their preduals 
map  Gaussian states to Gaussian states. They were first introduced and their 
generators were shown to be of the Lindblad type by Vanheuverzwijn \cite{19}. 
They were recently investigated in the context of quantum information theory by 
Heinosaari, Holevo and Wolf \cite{7}. Here we present the exact noisy 
Schr\"odinger equation which dilates such a semigroup to a quantum Gaussian 
Markov process.}


\section{Introduction}
\label{sec:1}

Consider a system whose state at any time $t$ is described by $n$  real 
coordinates 
$(\xi_1 (t), \xi_2(t), \ldots, \xi_n (t)).$ As an example  one may look at the 
system of a single particle moving in the space $\mathbb{R}^3$ and its state 
consisting of six coordinates, three for its position and three for its 
velocity components. Suppose a characteristic of the system described by a 
function $f$ of the state is being studied. Then such a characteristic changes 
with time and we write

\begin{equation}
X (t) = f (\xi_1 (t), \ldots, \xi_n (t)). \label{eq1.1}
\end{equation}
Such a change is described by its differential
\begin{equation}
d X (t) = \sum_{j=1}^{n} \,\frac{\partial f}{\partial \xi_j} 
(\xi_1 (t), \ldots, \xi_n (t)) \,d \xi_j(t).    \label{eq1.2}
\end{equation}
If $Y(\cdot)$ is another characteristic such that
\begin{equation}
Y(t) = g (\xi_1 (t), \ldots, \xi_n (t))    \label{eq1.3}
\end{equation}
then
\begin{equation}
d X (t) Y(t) = X (t) d Y(t) + Y (t) d X (t).    \label{eq1.4}
\end{equation}
Note that \eqref{eq1.2} is interpreted as
\begin{equation}
X (t+h) - X (t) = \int_t^{t+h} \,\,\sum_{j=1}^{n} \,\,\frac{\partial 
f}{\partial \xi_j} (\xi_1 (t), \ldots, \xi_n (t)) \,\frac{d \xi_j}{dt} \,dt     
\label{eq1.5}
\end{equation}
In two seminal papers in 1944 and 1951, K. Ito \cite{11}, \cite{12} developed a 
method for such a differential description when the path $$\bm{\xi}(t)=(\xi_1 
(t,\omega), \xi_2 (t, \omega), \ldots, \xi_n(t,\omega))$$ is random, i.e., 
subject to the laws of chance and the randomness is described by points 
$\omega$ in a probability space. Two paradigm examples of such random paths or 
trajectories come to our mind: the continous Brownian motion executed by a 
small particle suspended in a fluid and the jump motion of the number of 
radioactive particles emitted during the time interval $[0, t]$ by a 
radioactive substance undergoing radioactive decay. The first example yields a 
Gaussian process with independent increments and the second yields a Poisson 
jump process.

To begin with we consider a standard Brownian motion process $\bm{B}(t)  =  
(B_1(t),$ $ B_2 (t), \cdots, B_n(t))$ where $B_i$'s  are independent and each $\{ 
B_j(t), t \ge 0 \}$ is a standard Brownian motion process with continuous 
sample paths and independent increments, $B_j(0)=0$ and $B_j(t)$ being normally 
distributed with mean $0$ and variance $t.$ To do a differential analysis of 
functionals of such a Brownian motion one has to define stochastic integrals 
with respect to such Brownian paths which are known to be of
unbounded 
variation. Integrands in such a theory of stochastic integration are the {\it 
Ito functionals} or {\it nonanticipating functionals}. They are of the form 
$$f (t, \bm{B}(\cdot)) = f (t, \{ B (s), 0 \leq s \leq t \}), \quad 0 \leq t < 
\infty$$
which take real or complex values. In other words the random variable $f (t, 
\bm{B}(\cdot))$ depends on $t$ and the whole Brownian path in the interval 
$[0,t].$ For $n+1$ such Ito functionals $f_1, f_2, \ldots, f_n,$ $g$ the Ito 
theory associates an integral of the form
\begin{equation}
X (t) = X(0) + \sum_{j=1}^n \int_0^t f_j (s, \bm{B}) d B_j (s) + \int_0^t g (s, 
\bm{B}) ds   \label{eq1.6}
\end{equation}
and shows that for a large linear space of such vector-valued functionals 
$(f_1, f_2, \ldots, f_n,\, g),$ such integrals are well-defined and have many 
interesting properties. If \eqref{eq1.6} holds for every $t$ in an interval 
$[0,T]$ we write
\begin{equation}
d X (t) = \sum_{j=1}^n f_j (t,\bm{B}) d B_j (t) + g (t,\bm{B}) dt     
\label{eq1.7}
\end{equation}
for $0 \leq t \leq T$ with the prescribed  initial  value $X(0).$ If we have 
another relation of the form \eqref{eq1.7}, say,
\begin{equation}
d Y(t) = \sum_{j=1}^{n} h_j (t, \bm{B}) d B_j (t) + k (t, \bm{B}) dt   
\label{eq1.8}
\end{equation}
in $[0,T]$ with an initial value $Y(0),$ then both $X(t) = X(t, \bm{B}),$ $Y(t) 
= Y(t, \bm{B})$ are again Ito functionals and one can ask what is the 
differential of the product Ito functional $X(t) Y(t), 0 \leq t \leq T.$ The 
famous {\it Ito's formula} states that
\begin{equation}
d X(t) Y(t) = X(t) d Y(t) + Y(t) dX(t) + d X(t) dY(t)    \label{eq1.9}
\end{equation}
where
\begin{equation}
d X (t) d Y(t) = \left \{\sum_{j=1}^n \,f_j (t, \bm{B}) h_j (t, \bm{B}) \right 
\} dt.    \label{eq1.10}
\end{equation}
This is best expressed in the form of a multiplication table for the 
fundamental differentials $dB_j,$ $1 \leq j \leq n$ and $dt:$
\begin{equation}
\begin{tabular}{c|ccccc} \\ 
     & $d B_1$ & $dB_2$ & $\cdots$ & $dB_n$ & $dt$ \\ \hline
$dB_1$ & $dt$ & $0$ & $\cdots$ & $0$ & $0$ \\ 
$dB_2$ & $0$ & $dt$ & $\cdots$ & $0$ & $0$ \\
$\vdots$ & $\vdots$ & $\vdots$ & & $\vdots$ & $\vdots$ \\
$dB_n$ & $0$ & $0$ & $\cdots$ & $dt$ & $0$ \\
$dt$ & $0$ & $0$ & $\cdots$ & $0$  & $0$ 
\end{tabular}
    \label{eq1.11}
\end{equation}
Here the last diagonal entry and all the nondiagonal entries are $0.$ This 
multiplication table also implies that for any twice continuously 
differentiable function $\varphi (t, x_1, x_2, \ldots, x_n)$ on $\mathbb{R} 
\times \mathbb{R}^n,$ the Ito functional $X(t) = \varphi (t, B_1 (t), \ldots, 
B_n(t))$ satisfies
\begin{equation}
dX(t) = \sum_{j=1}^{n} \frac{\partial \varphi}{\partial x_j} (t, B_1 (t), 
\ldots, B_n (t)) d B_j (t) + \left \{\frac{\partial \varphi}{\partial t} + 
\frac{1}{2} \sum_{j=1}^n \frac{\partial^2 \varphi}{\partial x_j^2}  \right \} 
(t, \bm{B} (t)) dt.   \label{eq1.12}
\end{equation}
Formulae \eqref{eq1.9}-\eqref{eq1.12} constitute the backbone of the Ito 
stochastic calculus and its diverse applications.

Ito calculus has been extended to all local semimartingales and it is an 
extraordinarily rich theory with applications to many areas ranging from 
physics and biology to economics and social sciences. For the Poisson process 
$\{N_{\lambda} (t), t \ge 0 \}$ with intensity parameter $\lambda$ the 
multiplication table for differentials has the form
\begin{equation}
\begin{tabular}{c|cc}
 & $d N_{\lambda}$ & $dt$ \\    \hline
$dN_{\lambda}$ & $dN_{\lambda}$ & $0$ \\
$dt$ & $0$ & $0$
\end{tabular}
     \label{eq1.13}
\end{equation}
For a comprehensive account of classical stochastic calculus we refer to 
\cite{10}.

Coming to quantum theory we observe that both chance and noncommutativity of 
observables play an important role. A quantum stochastic process may be roughly 
described by a family $\{X_t\}$ of observables or, equivalently, selfadjoint 
operators in a complex and separable Hilbert space $\mathcal{H}$ together with 
a state $\rho$ which is a  nonnegative operator of unit trace in $\mathcal{H}.$ 
The operators $X_t$ at different time points may not commute with each other. 
However, one would like to have a `differential' description of $\{X_t\}$ in 
terms of the differentials of some fundamental processes which may be viewed as 
quantum analogues of processes like Brownian motion and Poisson process. Then 
the differential description  will depend on integrals of operator-valued 
processes with respect to the fundamental processes. Borrowing from the fact 
that there is a close connection between infinitely divisible distributions, 
classical stochastic processes with independent increments and free quantum 
fields on a boson Fock space as outlined in the papers of H. Araki \cite{1}, R. 
F. Streater \cite{18}, K. R. Parthasarathy and K. Schmidt \cite{17} we search 
for observables from free field theory to provide us the fundamental processes 
and their differentials. Our aim would then be to describe a quantum Ito's 
formula or an Ito table similar to \eqref{eq1.11} and \eqref{eq1.13}. This goal 
is achieved in the paper \cite{9} by Hudson and Parthasarathy. We shall follow 
\cite{9} and use it to construct examples of quantum Gaussian processes of the 
quasifree type \cite{19}, \cite{7} as solutions of quantum stochastic 
differential equations.

\section{Boson Fock space and Weyl operators}
\label{sec:2}
\setcounter{equation}{0}

All the Hilbert spaces we deal with will be assumed to be complex and separable 
and scalar products will be expressed in the Dirac notation. To any Hilbert 
space $\mathcal{H}$ we associate its {\it boson Fock space} 
$\Gamma(\mathcal{H})$ defined by
\begin{equation}
\Gamma (\mathcal{H}) = \mathcal{C} \oplus \mathcal{H} \oplus 
\mathcal{H}^{\textcircled{s}^{2}} \oplus \cdots \oplus 
\mathcal{H}^{\textcircled{s}\,^{n}} \oplus \cdots .   \label{eq2.1}
\end{equation}
where $\mathcal{C}$ is the $1$-dimensional Hilbert space of complex numbers and 
$\mathcal{H}^{\textcircled{s}^{n}}$ is the $n$-fold symmetric tensor 
product of copies of $\mathcal{H}.$ The subspace 
$\mathcal{H}^{\textcircled{s}^{n}}$ in $\Gamma(\mathcal{H})$ is called the 
$n$-{\it particle} subspace and $\mathcal{C}$ is called the {\it vacuum} 
subspace. For any 
$u \in \mathcal{H},$ define the {\it  exponential vector} $e(u)$ associated with 
$u$ by

\begin{equation}
e (u) = 1 \oplus u \oplus \frac{u^{\otimes^{2}}}{\sqrt{2!}}  \oplus \cdots 
\oplus \frac{u^{\otimes^{n}}}{\sqrt{n!}}  \oplus \cdots \label{eq2.2}
\end{equation}
and observe that
\begin{equation}
\langle e (u) | e (v) \rangle = \exp \langle u | v \rangle \quad \forall \,\, 
u, v \in \mathcal{H}.   \label{eq2.3}
\end{equation}
Denote by $\mathcal{E}$ the linear manifold generated by the set of all 
exponential vectors in $\Gamma (\mathcal{H})$ and call $\mathcal{E}$ the {\it 
exponential domain.} Then the following properties hold:
\begin{itemize}
 \item[(i)]\quad $\mathcal{E}$ is dense in $\Gamma(\mathcal{H});$

 \item[(ii)]\quad Any finite set of exponential vectors is linearly independent;

 \item[(iii)]\quad The correspondence $u \rightarrow e (u)$ is strongly 
continuous.

 \item[(iv)] \quad Any bounded operator $A$ in $\Gamma(\mathcal{H})$ is 
uniquely determined by the map \qquad \qquad $u \rightarrow A e (u),$ $u \in 
\mathcal{H}.$

 \item[(v)]\quad Any correspondence $e (u) \rightarrow A e(u),$ $u \in 
\mathcal{H}$ extends to a densely defined operator with the dense domain 
$\mathcal{E}.$

 \item[(vi)]\quad $e(0) = 1 \oplus 0 \oplus 0 \oplus \cdots$ is the vacuum 
vector and 
\begin{equation}
\psi (u) = e^{-\|u\|^2/2} e(u) \label{eq2.4}
\end{equation}
is a unit vector. The pure state with density operator $|\psi(u)\rangle \langle 
\psi (u)|$ is called the {\it coherent state} associated with $u.$

\item[(vi)] \quad The correspondence
\begin{equation}
W(u) e(v) = e^{-1/2 \|u\|^{2} - \langle u|v \rangle} e (u+v) \quad \forall \,\, 
v     \label{eq2.5}
\end{equation}
is scalar product preserving on the set of all exponential vectors for every 
fixed $u$ and therefore extends uniquely to a unitary operator $W(u)$ on 
$\Gamma (\mathcal{H}).$ The operator $W(u)$ is called the {\it Weyl operator} 
associated with $u.$
\end{itemize}

The Weyl operators play a central role in quantum theory and particularly in 
quantum stochastic calculus. They obey the following multiplication relations:
\begin{eqnarray}
 W(u) W(v) &=& e^{-i \,\,\Im \langle u|v \rangle} W (u+v), \label{eq2.6}\\
W(u) W(v) W(u)^{-1} &=& e^{-2i \,\Im \langle u|v \rangle} W(v) \label{eq2.7}
\end{eqnarray}
for all $u,v$ in $\mathcal{H}.$ The correspondence $u \rightarrow W(u)$ is 
strongly continuous. There is no proper subspace of $\Gamma(\mathcal{H})$ 
invariant under all the Weyl operators. We summarize by saying that the 
correspondence $u \rightarrow W(u)$ is a projective unitary and irreducible 
representation of the additive group $\mathcal{H}.$ Furthermore, the 
correspondence $t \rightarrow W(tu),$ $t \in \mathbb{R}$ is a strongly 
continuous unitary representation of $\mathbb{R}$ and hence by Stone's 
theorem there exists a unique selfadjoint operator $p(u)$ such that
\begin{equation}
W(tu) = e^{-it \,p(u)}, t \in \mathbb{R}, u \in \mathcal{H}. \label{eq2.8}
\end{equation}

If $\mathcal{H}$ is finite dimensional and $u \rightarrow W^{\prime}(u)$ is a 
strongly continuous map from $\mathcal{H}$ into the unitary group of a Hilbert 
space $K$ obeying the multiplication relations \eqref{eq2.6} with 
$W$ replaced 
by $W^{\prime}$ then there exists a Hilbert space ${h}$ and a Hilbert space 
isomorphism $U:K \rightarrow \Gamma (\mathcal{H}) \otimes h$ such 
that
$$U W^{\prime}(u) U^{-1} = W (u) \otimes I_h \quad \forall \,\,u \in 
\mathcal{H}. $$
In particular, if $W^{\prime}(\cdot)$ is also irreducible then $h$ is one 
dimensional and $W(\cdot)$ and $W^{\prime}(\cdot)$ are unitarily equivalent 
through $U.$ This is the well-known Stone-von Neumann theorem.

From \eqref{eq2.6} and \eqref{eq2.8} one obtains the following commutation 
relations:
\begin{equation}
[p(u), p(v)] = 2 i \,\Im \langle u|v \rangle \quad \forall \,\, u, v \in 
\mathcal{H} \label{eq2.9}
\end{equation}
on the domain $\mathcal{E}.$ The domain $\mathcal{E}$ is a common core for all 
the selfadjoint operators $p(u)$ and  $\mathcal{E}$ is contained in the domain 
of all products of the form $p(u_1) p(u_2) \ldots p(u_n),$ $u_i \in 
\mathcal{H}.$ Writing
\begin{eqnarray}
 q(u) &=& -p (iu) \label{eq2.10} \\
a(u) &=& \frac{1}{2} (q(u) + i p (u)) \label{eq2.11}\\
a^{\dagger} (u) &=& \frac{1}{2} (q(u) - i p (u)) \label{eq2.12}
\end{eqnarray}
we obtain closable operators in the domain $\mathcal{E}$ obeying the following 
commutation relations:
\begin{eqnarray}
 [a(u), a(v)] &=& 0, \\ \nonumber 
[a^{\dagger} (u), a^{\dagger} (v)] &=& 0,\\ \nonumber 
[a(u), a^{\dagger} (v)] &=& \langle u | v \rangle. \label{eq2.13}
\end{eqnarray}
Furthermore the correspondence $u \rightarrow a (u)$ is antilinear, $u 
\rightarrow a^{\dagger} (u)$ is linear,
\begin{eqnarray*}
 a (u) e(v) &=& \langle u | v \rangle e (v), \\
\langle e(v)| a(u) e(w) \rangle &=& \langle a^{\dagger} (u) e(v) | e(w) \rangle
\end{eqnarray*}
for all $u, v, w$ in $\mathcal{H}.$ If we denote the closures of $a(u)$ and 
$a^{\dagger} (v)$ on $\mathcal{E}$ by the same symbols one obtains their 
actions on $n$-particle vectors as follows:
\begin{eqnarray}
 a(u) e(0) &=& 0, \\ \nonumber
a(u) v^{\otimes^{n}} &=& \sqrt{n} \langle u | v \rangle v^{\otimes^{n-1}} \\ 
\nonumber 
a^{\dagger} (u) v^{\otimes^{n}} &=& (n+1)^{-1/2} \sum_{r=0}^{n} v^{\otimes^{r}} 
\otimes u \otimes v^{\otimes^{n-r}}. \label{eq2.14}
\end{eqnarray}
These relations show that $a(u)$ maps $\mathcal{H}^{\textcircled{s}^{n}}$ into 
$\mathcal{H}^{\textcircled{s}^{n-1}}$ whereas $a^{\dagger} (u)$ maps 
$\mathcal{H}^{\textcircled{s}^{n}}$ into $\mathcal{H}^{\textcircled{s}^{n+1}}.$ 
In view of this property $a(u)$ and $a^{\dagger}(u)$ are called {\it 
annihilation} and {\it creation} operators associated with $u.$

For any selfadjoint operator $H$ in $\mathcal{H}$ with domain $D(H)$ denote by 
$\mathcal{E}(D(H))$ the linear manifold generated by $\{e(u) | u \in D (H)\}.$ 
For any unitary operator $U$ in $\mathcal{H}$ define the operator $\Gamma(U)$ 
in $\Gamma (\mathcal{H})$ by putting
$$\Gamma (U) e(u) = e (Uu), \quad u \in \mathcal{H},$$
noting that it is scalar product preserving and hence extending it to a unitary 
operator on $\Gamma (\mathcal{H}).$ Then for any two unitary operators $U,V$ in 
$\mathcal{H}$ one has the relation $\Gamma (U) \Gamma(V)=\Gamma(UV).$ Now we 
see that $\Gamma (e^{-it H}),$ $t \in \mathbb{R}$ is a strongly continuous  one 
parameter unitary group and therefore, by Stone's theorem, there exists a 
selfadjoint operator $\lambda (H)$ in $\Gamma (\mathcal{H})$ such that
\begin{equation}
\Gamma (e^{-itH}) = e^{-it \,\,\lambda (H)}, t \in \mathbb{R}.   \label{eq2.15}
\end{equation}
$\Gamma(U)$ is called the {\it second quantization} of $U$ and $\lambda(H)$ is 
called the {\it differential second quantization} of the selfadjoint operator 
$H.$ For $v \in D(H)$ one has the relation
\begin{equation}
\lambda (H) v^{\otimes^{n}} = \sum_{r=0}^{n-1} v^{\otimes^{r}} \otimes \lambda 
(H)  v \otimes v^{\otimes^{n-r-1}}   \label{eq2.16}
\end{equation}
Thus $\lambda(H)$ sends an $n$-particle vector to an $n$-particle vector and 
hence $\lambda(H)$ is called the {\it conservation operator} associated with 
$H.$

If $H$ is any bounded operator in $\mathcal{H}$ one can express
$$H = \frac{H+H^{\dagger}}{2} + i \frac{H-H^{\dagger}}{2i}$$
and put
$$\lambda (H) = \lambda \left (\frac{H+H^{\dagger}}{2} \right ) + i \lambda 
\left ( \frac{H-H^{\dagger}}{2i} \right ).$$
If we denote by $\mathcal{B}(\mathcal{H})$ the algebra of all bounded operators 
on $\mathcal{H}$ with the adjoint operation $\dagger$ as involution then the 
following hold:
\begin{itemize}
 \item[(i)] \quad The map $H \rightarrow \lambda(H)$ is linear on 
$\mathcal{B}(\mathcal{H});$
 \item[(ii)] \quad $\lambda (H)^{\dagger} = \lambda (H^{\dagger});$
 \item[(iii)] \quad On the domain $\mathcal{E}$ the following commutation 
relations hold:
\begin{eqnarray}
{\rm (a)} &&   [\lambda (H_1), \lambda (H_2)] = \lambda ([H_1, 
H_2]),  \\ \nonumber
{\rm (b)} && [a(u), \lambda(H)] = a (H^{\dagger} u ), \\ \nonumber 
{\rm (c)} && [\lambda (H), a^{\dagger} (u)] = a^{\dagger} (Hu)    
\label{eq2.17}
\end{eqnarray}
\end{itemize}
In terms of the Weyl operators and second quantization of unitary operators in 
$\mathcal{H}$ one has the following relations:
\begin{eqnarray*}
 W(u) \Gamma (U) &=& e^{-\frac{1}{2} \|u\|^{2}} e^{a^{\dagger} (u)} \Gamma (U) 
e^{-a(U^{-1} u)},\\
W(u) &=&   e^{a^{\dagger}(u)-a(u)} \\
\Gamma(U) W(u) \Gamma(U)^{-1} &=& W (Uu)
\end{eqnarray*}
for all $u \in \mathcal{H}$ and unitary operators $U$ in $\mathcal{H}.$ The 
first 
two identities are to be interpreted in a weak sense on the domain 
$\mathcal{E}.$

It is to be emphasized that the Weyl operators and second quantization yield a 
rich harvest of interestig observables $p(u)$ and $\lambda(H)$ as $u$ varies in 
$\mathcal{H}$ and $H$ varies in the set of all selfadjoint operators. In the 
next section we shall examine their statistical properties in the vacuum state.

\section{Statistics of observables arising from Weyl operators and second 
quantization}
\label{sec:3}
\setcounter{equation}{0}

First, we begin with a general remark about the mechanism by which classical 
stochastic processes arise in the quantum framework. Suppose $\mathcal{H}_0$ is 
the Hilbert space of a quantum system and $\rho$ is a state in $\mathcal{H}_0,$ 
i.e., a positive operator of unit trace. Let $\{X_t, t \in T\}$ be a commuting 
family of selfadjoint operators in $\mathcal{H}_0$ or, equivalently, 
observables. Then write
\begin{equation}
\varphi_{t_{1}, t_{2}, \ldots, t_{n}} (x_1, x_2, \ldots, x_n) = \Tr \,\rho \, 
\exp \,\,i \sum_{j=1}^{n} x_j X_{t_{j}} \label{eq3.1}
\end{equation}
for any $x_1, x_2, \ldots, x_n$ in $\mathbb{R}$ and $\{t_1, t_2, \ldots, t_n\} 
\subset T.$ Then $\varphi_{t_{1}, t_{2}, \ldots, t_{n}}$ is the characteristic 
function (or Fourier transform of a probability measure $\mu_{t_{1}, t_{2}, 
\ldots, t_{n}}$ in $\mathbb{R}^n$ or, more precisely, $\mathbb{R}^{\{t_{1}, 
t_{2}, \ldots, t_{n}\}},$ so that
\begin{equation}
\varphi_{t_{1}, t_{2}, \ldots, t_{n}} (x_1, x_2, \ldots, x_n) = 
\int_{\mathbb{R}^{n}} e^{i \bm{x} \cdot \bm{y}} \mu_{t_{1}, \ldots, t_{n}} 
(d\bm{y}).     \label{eq3.2}
\end{equation}
The family $\{ \mu_{t_{1}, \ldots, t_{n}}, \{ t_1, t_2, \ldots, t_n\} \subset T 
\}$ of finite dimensional probability distributions is consistent in the sense 
of Kolmogorov and therefore, by Kolmogorov's consistency theorem, there exists 
a unique probability measure $\mu$ on the product Borel space $\mathbb{R}^T$ 
whose projection on any $\mathbb{R}^{t_{1}, \ldots, t_{n}}$ is $\mu_{t_{1}, 
\ldots, t_{n}}.$ In other words, one obtains a classical stochastic process 
described by $\mu.$ If $T$ is a real vector space and the map $t \rightarrow 
X_t$ is linear we see that, in any state $\rho,$ one obtains a classical random 
field over $T.$

Now we choose a Hilbert space $\mathcal{H}$ and fix an orthonormal basis 
$\{e_1, e_2, \ldots\}$ in $\mathcal{H}.$ Define $\mathcal{A}_{\mathbb{R}}$ to 
be the real linear space of all bounded selfadjoint operators in $\mathcal{H}$ 
and fix a subspace $\mathcal{A}_0 \subset \mathcal{A}_{\mathbb{R}}$ where any 
two elements of $\mathcal{A}_0$ commute with each other. In the boson Fock 
space $\Gamma (\mathcal{H})$ defined by \eqref{eq2.1} in Section 2 consider the 
following three families of observables:

\begin{itemize}
 \item[(i)] \quad $\{p(u) | u \in \mathcal{H}_{\mathbb{R}} \},$
 \item[(ii)] \quad $\{q (u) | u \in \mathcal{H}_{\mathbb{R}} \},$
 \item[(iii)] \quad $\{\lambda (H) | H \in \mathcal{A}_0 \},$\\
where $\mathcal{H}_{\mathbb{R}}$ is the closed real linear subspace spanned by 
$e_i, e_2, 
\ldots.$
\end{itemize}

Then each of the three families above is commutative, thanks to \eqref{eq2.8}, 
\eqref{eq2.9}, \eqref{eq2.10}, \eqref{eq2.15} and (17) 
in Section 2. Furthermore, 
the correspondence $u \rightarrow p(u),$ $u \rightarrow q(u)$ and $H 
\rightarrow \lambda (H)$ are all real linear. Of course, one must take care of 
the fact that the operators involved are unbounded. By our general remarks at 
the beginning of this section each of the three families (i)-(iii) of 
observables determines a classical random field in any fixed state $\rho$ in 
$\Gamma(\mathcal{H}).$

We now specialize to the case when $\rho = |\psi (u_0) \rangle \langle \psi 
(u_0) |$ is the coherent state corresponding to $u_0$ defined by \eqref{eq2.2} 
and \eqref{eq2.4} in Section 2. Note that
\begin{eqnarray*}
 \lefteqn{\langle \psi (u_0) | e^{-it p(u)} | \psi (u_0) \rangle} \\
&=& \langle \psi(u_0) | W (tu) | \psi (u_0) \rangle \\
&=& e^{2it \,\Im \langle u_0 | u \rangle - \frac{1}{2} t^{2} \|u\|^{2}}
\end{eqnarray*}
 which is the characteristic function of the normal distribution with mean $2 
\,\Im \langle u_0 | u \rangle$ and variance $\|u\|^2.$ This also implies that
\begin{eqnarray*}
\lefteqn{\Tr \,e^{-i \sum\limits_{j=1}^{n} t_{j} p (u_{j})} |\psi (u_0) \rangle 
\langle \psi (u_{0}) |  }\\
&=& \exp \left \{2 i \sum_{j=1}^{n} t_j \Im \langle u_0 | u_j \rangle - 
\frac{1}{2} \sum_{i,j} t_i t_j \langle u_i | u_j \rangle \right \}.
\end{eqnarray*}
In other words $\{p(u) | u \in \mathcal{H}_{\mathbb{R}}\}$ is a classical 
Gaussian random field in the state $|\psi (u_0) \rangle \langle \psi (u_0)|$ 
with mean functional $m(\cdot)$ and covariance kernel $K(.,.)$ given by
$$m(u) = 2 \Im \langle u_0 | u \rangle, K(u,v) = \langle u|v \rangle$$
for all $u, v \in \mathcal{H}_{\mathbb{R}}.$ By the same arguments $\{q (u) | 
u \in \mathcal{H}_{\mathbb{R}}\}$ is again a classical Gaussian random field 
with the same covariance kernel but its mean functional $m^{\prime}$ is given 
by 
$m^{\prime} (u) = 2 \Re \langle u_0|u\rangle.$ Note that $[q(u), p(v)]=2i 
\langle u | v \rangle$ for all $u, v \in \mathcal{H}_{\mathbb{R}}$ and 
therefore $q(u)$ and $p(v)$ need not commute with each other.

Let us now examine the case of the third family. We have for any fixed $u \in 
\mathcal{H}$ and any selfadjoint operator $H$ in $\mathcal{H}$
\begin{eqnarray*}
\lefteqn{ \langle \psi (u)| e^{-it \lambda (H) }    | \psi (u) \rangle   } \\
&=& e^{-\|u\|^{2}} \langle e(u) | \Gamma (e^{-itH}) | e (u) \rangle   \\
&=&  e^{-\|u\|^{2}} \langle e(u) | e (e^{-itH}u) \rangle  \\
&=&    \exp \langle u | e^{-itH} - 1 | u \rangle.
\end{eqnarray*}
If $P^H$ is the spectral measure of $H$ so that
$$H = \int_{\mathbb{R}} x P^H (dx)$$
then
\begin{eqnarray*}
\lefteqn{ \langle \psi (u) | e^{-it \lambda(H)} | \psi (u) \rangle       } \\
&=& \exp \int (e^{itx} -1) \langle u | P^H (dx) | u \rangle.
\end{eqnarray*}
In other words the distribution of the observable $\lambda (H)$ in the coherent 
state $|\psi (u) \rangle \langle \psi (u)|$ is the infinitely divisible 
distribution with characteristic function having the L\'evy-Khinchine 
representation
$$\exp \int (e^{itx} -1) \mu_{H,u} (dx) $$
where
$$\mu_{H,u} (E) = \langle u | P^H (E) | u \rangle $$
for any Borel set $E \subset \mathbb{R}.$ Thus $\{\lambda (H) | H \in 
\mathcal{A}_0\}$ in the coherent state $| \psi (u) \rangle \langle \psi (u) |$ 
realizes a classical random field over the vector space $\mathcal{A}_0 \subset 
\mathcal{A}_{\mathbb{R}}$ for which $\lambda (H)$ has the infinitely divisible 
distribution with L\'evy measure $\mu_{H, u}$ described above.

We observe that $\{p(u), u \in \mathcal{H}\}$ is a real linear but 
noncommutative family of observables such that, in a fixed coherent state, each 
$p(u)$ has a Gaussian distribution. It is natural to call the pair $\{ \{ p(u), 
u \in \mathcal{H}\}, |\psi(u_0) \rangle \langle \psi(u_0)| \}$ a {\it quantum 
Gaussian field}. Similarly, $\{\lambda (H) | H \in \mathcal{A}_{\mathbb{R}}\}$ 
where $\mathcal{A}_{\mathbb{R}}$ is the real linear space of all bounded 
selfadjoint operators in $\mathcal{H},$ is a real linear but noncommutative 
space of observables such that in the state $|\psi (u)\rangle \langle 
\psi(u)|,$ $\lambda (H)$ has the infinitely divisible distribution with L\'evy 
measure $\langle u | P^H (\cdot)|u \rangle$ for every $H.$ Thus it is natural 
to call the pair $\{\{\lambda (H) | H \in \mathcal{A}_{\mathcal{R}} \}, 
|\psi(u) \rangle \langle \psi (u)|  \}$ a {\it quantum L\'evy field.} It looks 
like an interesting problem to examine the nature of the most general quantum 
Gaussian and L\'evy fields.

Our discussions in this section also show that the observable fields $\{p(u), 
u \in \mathcal{H} \}$ and $\{\lambda (H), H \in \mathcal{A}_{\mathbb{R}}\}$ 
constitute a natural ground for constructing a theory of stochastic integration 
and paving the way for introducing a quantum stochastic calculus.

\section{$G$-stationary quantum Gaussian processes}
\label{sec:4}
\setcounter{equation}{0}

Let $G$ be a second countable metric group acting continuously on a second 
countable metric space $A$ and let $K(\alpha, \beta),$ $\alpha, \beta \in A$ be 
a continuous positive definite $G$-invariant kernel so that
\begin{equation}
 K (g \alpha, g \beta) = K(\alpha, \beta) \quad \forall \,\alpha, \beta \in A, 
g \in G.  \label{eq4.1}
\end{equation}
Then by the GNS principle there exists a Gelfand triple $(\mathcal{H}, \lambda, 
\pi)$ consisting of a Hilbert space $\mathcal{H},$ a map $\lambda:A \rightarrow 
\mathcal{H}$ and a strongly continuous unitary representation $g \rightarrow 
\pi(g)$ of $G$   such that the following 
properties hold:
\begin{itemize}
 \item[(i)]\qquad $\lambda$ is continuous and $\mathcal{H}$ is the closed 
linear span of $\{\lambda (\alpha), \alpha \in A\};$
\item[(ii)] \qquad $K (\alpha, \beta) = \langle \lambda (\alpha) | \lambda 
(\beta)\rangle \quad \forall\,\, \alpha, \beta \in A;$
\item[(iii)] \qquad $\pi (g) \lambda (\alpha) = \lambda (g \alpha) \quad \forall 
\,\, g \in G, \alpha \in A;$
 \item[(iv)]\qquad  If there is another triple $(\mathcal{H}^{\prime}, 
\lambda^{\prime}, \pi^{\prime})$ satisfying (i), (ii) and (iii) then there 
exists a Hilbert space isomorphism $U:\mathcal{H} \rightarrow 
\mathcal{H}^{\prime}$ such that $U \lambda (\alpha) = \lambda^{\prime} (\alpha) 
\forall \alpha$ and $U \pi (g) U^{-1} = \pi^{\prime} (g) \forall \,\,g \in G.$ 
\\
For a proof we refer to \cite{17}.
\end{itemize}

Now consider the boson Fock space $\Gamma (\mathcal{H})$ where $(\mathcal{H}, 
\lambda, \pi)$ is a Gelfand triple described above and associated  with 
$K(.,.).$ Define the observables
\begin{eqnarray*}
 q(\alpha) &=& \frac{1}{\sqrt{2}} (a (\lambda (\alpha)) + a^{\dagger} 
(\lambda(\alpha))), \\
p(\alpha) &=& \frac{1}{i \sqrt{2}} (a (\lambda (\alpha))-a^{\dagger} 
(\lambda (\alpha))
\end{eqnarray*}
where $a(u)$ and $a^{\dagger} (u)$ are the annihilation and creation operators 
associated with $u \in \mathcal{H}.$ Then

\begin{eqnarray*}
\left [q (\alpha), q (\beta) \right ] &=& \left [p(\alpha), p(\beta) \right ]=i 
\,\Im \,K(\alpha,\beta) \\
\left [q(\alpha), p(\beta)\right ] &=& i \,\Re\, K(\alpha, \beta)
\end{eqnarray*}
for all $\alpha, \beta \in A.$ We shall now examine the probability 
distribution of an arbitrary finite real linear combination of the form
$$Z = \sum_j \left ( x_j q (\alpha_j) + y_j p(\alpha_j) \right ) $$
where $\alpha_j \in A,$ and $x_j, y_j$ are real scalars. If we write $z_j = x_j 
+ i y_j$ then
$$itZ = a^{\dagger} \left (it\,\frac{\sum z_j \lambda (\alpha_j)}{\sqrt{2}} 
\right ) - a \left ( \frac{it \, \sum z_j \lambda (\alpha_j)}{\sqrt{2}}\right 
)$$
and
$$e^{itZ} = W \left ( \frac{it}{\sqrt{2}} \sum_j z_j \alpha_j \right )$$
is the unitary Weyl operator for any $t \in \mathbb{R}.$ Thus
$$\langle e(0) | e^{itZ} | e(0) \rangle = e^{-\frac{t^2}{4} \sum\limits_{j,k} 
\bar{z}_j z_k 
K (\alpha_{j}, \alpha_{k})}.$$
This shows that, in the vacuum state, the observable $Z$ has the normal 
distribution $N (0, \frac{1}{2} \sum_{j,k} \bar{z}_j z_k K (\alpha_j, 
\alpha_k)).$

Furthermore
$$\Gamma (\pi (g)) \left [\begin{array}{c} q(\alpha) \\ p(\alpha) \end{array} 
\right ] \Gamma (\pi(g))^{-1} = \left [\begin{array}{c} q(g\alpha) \\ 
p(g \alpha) \end{array} 
\right ], $$
where $\Gamma (\pi(g))$ is the second quantization of $\pi(g).$ Thus the 
observable $Z$ and $\Gamma (\pi(g)) Z \,\Gamma (\pi (g))^{-1}$ have the same 
normal distribution.

This shows that the family of observables $\{ q(\alpha), p(\alpha), \alpha \in 
A \}$ constitute a $G$-invarinat quantum Gaussian process in the vacuum state.

If $K$ is a real positive definite kernel then each of the families $\{q 
(\alpha), \alpha \in A\}$ and $\{p(\alpha), \alpha \in A\}$ is commutative and 
therefore each of them executes a classical $G$-stationary Gaussian process.

\section{Quantum stochastic calculus and a noisy Schr\"odinger equation}
\label{sec:5}
\setcounter{equation}{0}

Consider a quantum system $S$ whose states are described by density operators 
in a Hilbert space $\mathcal{H}_S.$ Suppose that this system is coupled to a 
bath or an extermal environment whose states are described by density operators 
in a boson Fock space of the form
\begin{equation}
\mathcal{H} = \Gamma (L^2 (0, \infty) \otimes \mathbb{C}^d)   \label{eq5.1}
\end{equation}
where $e_i = (0, \ldots, 0, 1,0\ldots,0)$ with $1$ in the $i$-th place, 
$i=1,2,\ldots, d$ is fixed as a canonical orthonormal basis in $\mathbb{C}^d.$ 
We write
\begin{eqnarray}
 \widetilde{\mathcal{H}} &=& \mathcal{H}_S \otimes \mathcal{H} \label{eq5.2}\\
 \widetilde{\mathcal{H}}_t &=& \mathcal{H}_S \otimes \Gamma (L^2 [0,t] \otimes 
\mathbb{C}^d) \label{eq5.3} \\
 \widetilde{\mathcal{H}}^t &=& \Gamma (L^2 ([t, \infty] \otimes \mathbb{C}^d) 
\label{eq5.4} \\
 \widetilde{\mathcal{H}}^{(s,t)} &=& \Gamma (L^2 [s,t] \otimes \mathbb{C}^d), 0 
\leq s \leq t < \infty \label{eq5.5}
\end{eqnarray}
Then for any $0 < t_1 < t_2 < \cdots < t_n < \infty,$
$$\mathcal{H}^{(0,\infty)} = \mathcal{H}^{(0,t_{1})} \otimes 
\mathcal{H}^{(t_{1}, t_{2})} \otimes \cdots \otimes \mathcal{H}^{(t_{n-1}, 
t_{n})} \otimes \mathcal{H}^{t_{n}}$$
where the identification of both sides can be achieved through products of 
exponential vectors from different components. The noise accumulated from the 
bath during the time period $(s,t)$ admits a description through observables in 
the Hilbert space $\mathcal{H}^{(s,t)}.$ This also suggests that noise can be 
described by observables in a general continuous tensor product Hilbert space 
but the boson Fock space $\mathcal{H}$ is the simplest such model in which 
$\mathbb{C}^d$ means that there are $d$ degrees of freedom in the selection of 
noise.

We introduce the following family of noise processes:
\begin{eqnarray}
\Lambda_0^i (t) &=& I_S \otimes a \left (1_{[0,t]}  \otimes e_i \right ) 
\label{eq5.6} \\
\Lambda^0_i (t) &=& I_S \otimes a^{\dagger} \left (1_{[0,t]}  \otimes e_i 
\right 
)  \label{eq5.7} \\
\Lambda_j^i (t) &=& I_S \otimes  \lambda (I_{[0,t]} \otimes |e_j \rangle 
\langle e_i|),  \label{eq5.8} \\
\Lambda_0^0 (t) &=& t \,\,I_{\widetilde{\mathcal{H}}} \label{eq5.9}
\end{eqnarray}
where the indices $i,j$ vary in $\{1,2, \ldots, d\},$ $1_{[0,t]}$ is the 
indicator function of $[0,t]$ as an element of $L^2 ([0,\infty]),$ $1_{[0.t]}$ 
is the operator of multiplication by $1_{[0,t]}$ in $L^2 ([0, \infty))$ and 
$I_S,$ $I_{\widetilde{\mathcal{H}}}$ denote respectively the identity operators 
in $\mathcal{H}_S,$ $\widetilde{\mathcal{H}}.$ We shall use Greek letters 
$\alpha, \beta, \ldots$ to indicate indices in $\{0,1,2,\ldots,d\}.$ All the 
operators $\Lambda_{\beta}^{\alpha}(t)$ in \eqref{eq5.6}-\eqref{eq5.9} are 
well-defined on the linear manifold generated by elements of the form $\psi 
\otimes e (f),$ $\psi \in \mathcal{H}_S$ and $f \in L^2 (\mathbb{R}_{+}) 
\otimes \mathbb{C}^d,$  a Hilbert space which can be viewed as the space of 
$\mathbb{C}^d$-valued functions on $\mathbb{R}_{+}$ which are norm square 
integrable. We interpret $\{\Lambda_0^0 (t)\}$ as {\it time process}, 
$\{\Lambda_0^i (t) \}$ and $\{\Lambda^0_i (t) \}$ as {\it annihilation} and 
{\it creation processes} of {\it type} (or {\it colour}) $i$ and       
$\{\Lambda_j^i (t)\}$ as a {\it conservation process} 
which {\it changes the type $i$ to type $j.$} It is interesting to note that, 
for $0<s<t<\infty,$ $\Lambda_{\beta}^{\alpha} (t) - 
\Lambda_{\beta}^{\alpha} (s)$ 
as an operator is active only in the sector $\mathcal{H}^{(s,t)}$ of the 
continuous tensor product $\mathcal{H}$ defined in \eqref{eq5.1}-\eqref{eq5.5}. 
The processes $\{\Lambda_{\beta}^{\alpha} (t)\}$ will be the fundamental 
processes with respect to which stochastic integrals can be defined.

A family $X = \left \{ X(t), 0 \leq t < \infty \right \}$ of operators in 
$\widetilde{\mathcal{H}}$ is said to be {\it adapted} if, for each $t,$ there 
exists an operator $X_t$ in $\widetilde{\mathcal{H}}_t$ such that $X(t) = X_t 
\otimes I^t$ where $I^t$ is the identity operator in $\mathcal{H}^t.$ An 
adapted process $X$ is said to be {\it simple} if there exists a partition $0 < 
t_1 < t_2 < \cdots < t_n < \cdots < \infty$ of $[0, \infty)$ so that $t_n 
\rightarrow \infty$ as $n \rightarrow \infty$ and
\begin{equation}
X(t) = \sum_{j=0}^{\infty} X (t_j) 1_{[t_{j}, t_{j+1})} (t) \quad \forall \,\, 
t   \label{eq5.10}
\end{equation}
where $t_0 =0.$ If $X_1, X_2, \ldots, X_k$ are simple adapted processes it is 
clear that by taking the intersection of their respective partitions one can 
express $X_j$ in the form \eqref{eq5.10} in terms of a partition independent of 
$j.$

Now suppose $\{E_{\beta}^{\alpha}\}$ is a family of simple adapted processes 
with respect to a single partition $0 < t_1 < t_2 < \cdots .$ Then define the 
(quantum) stochastic integral
\begin{equation}
\int_0^t E_{\beta}^{\alpha} (s) d \Lambda_{\alpha}^{\beta} (s) = \sum_j 
E_{\beta}^{\alpha} (t_j) \left \{ \Lambda_{\beta}^{\alpha} (t_{j+1} \wedge 
t) - 
\Lambda_{\beta}^{\alpha} (t_j \wedge t) \right \}   \label{eq5.11}
\end{equation}
where we adopt the Einstein convention that repeated Greek index means 
summation with 
respect to that index. Here $a \wedge b$ denotes the minimum of $a$ and $b.$ 

The stochastic integral \eqref{eq5.11} as a function of $t$ is again an adapted 
process. Quantum stochastic calculus tells us that the definition 
\eqref{eq5.11} 
can be completed to extend the notion of stochastic integral to a rich class of 
adapted processes. For details we refer to the paper by R. L. Hudson and K. R. 
Parthasarathy (1984) and also the book \cite{14}. With this completed 
definition 
we consider adapted processes of the form
\begin{equation}
X(t) = X(0) + \int_0^t E_{\beta}^{\alpha} (s) d \Lambda_{\alpha}^{\beta} (s)    
\label{eq5.12}
\end{equation}
where $X(0) = X_0 \otimes I$ with $X_0$ an operator in $\mathcal{H}_S$ and $I,$ 
the identity operator in $\mathcal{H}.$ We express \eqref{eq5.12} as
\begin{equation}
d X (t) =  E_{\beta}^{\alpha} (t) d \Lambda^{\beta}_{\alpha} (t)  \label{eq5.13}
\end{equation}
with initial value $X_0.$ It is interesting to note that 
$E_{\beta}^{\alpha}(t)$ is active in $\widetilde{\mathcal{H}}_t$ whereas $d 
\Lambda_{\alpha}^{\beta} (t)$ is active in $\mathcal{H}^{(t, t + dt)}$ and 
therefore $E_{\beta}^{\alpha}(t)$ and $d \Lambda_{\alpha}^{\beta}(t)$ commute 
with each other.

Now suppose $Y(t)$ is another adapted process with initial value $Y_0$ and
\begin{equation}
 d Y(t) = F_{\beta}^{\alpha} (t) d \Lambda_{\alpha}^{\beta} (t). \label{eq5.14}
\end{equation}
Then it is natural to ask, in the light of our initial discussion in Section 
1, what is the differential of $X(t)Y(t).$ The answer to this question was 
provided by Hudson and Parthasarathy \cite{9}: In fact 
\begin{equation}
d (X(t)Y(t)) = X (t) dY(t) + (dX(t)) Y(t) + d X(t) d Y (t)   \label{eq5.15}
\end{equation}
where, using the Einstein convention, we have
\begin{eqnarray}
 X(t) d Y(t) &=& X(t) F_{\beta}^{\alpha} (t) d \Lambda_{\alpha}^{\beta} (t), 
\label{eq5.16}\\
 (dX(t)) Y(t) &=& E_{\beta}^{\alpha} (t) Y(t) d \Lambda_{\alpha}^{\beta} (t), 
\label{eq5.17}\\
 dX(t) dY(t) &=& E_{\beta}^{\alpha} (t) F_{\varepsilon}^{\gamma} (t) d 
\Lambda_{\alpha}^{\beta} (t).d \Lambda_{\gamma}^{\varepsilon} (t) 
\label{eq5.18} 
\\
 d \Lambda_{\alpha}^{\beta} (t) d \Lambda_{\gamma}^{\varepsilon}(t) &=& 
\widehat{\delta}_{\gamma}^{\beta} d \Lambda_{\alpha}^{\varepsilon} (t), 
\label{eq5.19}\\
 \widehat{\delta}_{\gamma}^{\beta} &=& \left \{\begin{array}{lcl} 
\delta_{\gamma}^{\beta} & \mbox{if}& \beta \neq 0, \gamma \neq 0 \\ 0 & 
\mbox{otherwise} & \end{array}  \right . \label{eq5.20}
\end{eqnarray}
Equations \eqref{eq5.15}-\eqref{eq5.20} completely describe how the classical 
Leibnitz formula gets corrected by the product of the differentials of $X(t)$ 
and $Y(t)$ in terms of the products of the differentials of the fundamental 
processes $\{\Lambda_{\beta}^{\alpha}, \alpha, \beta \in \{0,1,\ldots, d\}.$ 
This is the `{\it quantum Ito's formula}'. Equations \eqref{eq5.19} and 
\eqref{eq5.20} constitute the boson Fock quantum stochastic calculus version of 
the classical Ito's tables \eqref{eq1.11} and \eqref{eq1.13} of Section 1 for 
Brownian motion 
and Poisson process. If we write
\begin{equation}
Q_i (t) = (\Lambda_0^i (t) + \Lambda_i^0 (t))^{\sim}, 1 \leq i \leq d,   
\label{eq5.21}
\end{equation}
$\sim$ denoting closure over $\mathcal{E}$ then $(Q_1, Q_2, \ldots, Q_d)$ is a 
collection of commuting observables and they execute the $d$-dimensional 
Brownian motion process in the vacuum state  $|\psi(0) \rangle \langle \psi 
(0)|$ and equations \eqref{eq5.19} and \eqref{eq5.20} imply that $d Q_i (t) d 
Q_j (t) = \delta_{ij} dt,$ which is the classical Ito correction formula for 
Brownian motion.

If we define, for $\lambda_i > 0,$ $1 \leq i \leq d,$
\begin{equation}
N_i (t) = \sqrt{\lambda_i} Q_i (t) + \Lambda_i^i (t) + \lambda_i t   
\label{eq5.22}
\end{equation}
then $\{ N_i (t) \}$ is a commutative process which executes a classical 
Poisson process with intensity $\lambda_i$ in the Fock vacuum state and 
\begin{equation}
d N_i (t) d N_j (t) = \delta_{ij} d N_j (t)   \label{eq5.23}
\end{equation}
for any $i,j.$ This is Ito's formula for the classical Poisson process.

Suppose $L_{\beta}^{\alpha},$ $\alpha, \beta \in \{0,1, \ldots, d\}$ are 
bounded operators in the system Hilbert space $\mathcal{H}_S$ and
$$L_{\beta}^{\alpha} (t) = L_{\beta}^{\alpha} \otimes I_{\mathcal{H}} \quad 
\forall \,\, t \ge 0 $$
where $I_{\mathcal{H}}$ is the identity operator in the bath Hilbert space 
$\mathcal{H}.$ Thus we obtain a family of constant adapted processes. We now 
write a quantum stochastic differential equation of the exponential type:
\begin{equation}
d U(t) = \{   L_{\beta}^{\alpha} (t) d \Lambda_{\alpha}^{\beta}(t) \} U(t), t 
\ge 
0 \label{eq5.24}
\end{equation}
with the initial condition $U(0) = I.$ Then we have the following theorem:

\begin{theorem}[Hudson and Parthasarathy \cite{9}]\label{thm1}
Equation \eqref{eq5.24} with the initial condition $U(0)=I$ has a unique 
unitary operator-valued adapted process as a solution if and only if
\begin{equation}
 L_{\beta}^{\alpha} + (L^{\beta}_{\alpha})^{\dagger} 
+\sum_i(L^{i}_{\alpha})^{\dagger} 
L_{\beta}^{i} = L_{\beta}^{\alpha} + (L^{\beta}_{\alpha})^{\dagger} + \sum_i
L_{i}^{\alpha} (L^{\beta}_{i})^{\dagger} = 0 \label{eq5.25}
\end{equation}
for each $\alpha, \beta$ in $\{0, 1, \ldots, d\}.$
\end{theorem}

\begin{remark}\label{rem1}
Suppose $L_{\beta}^{\alpha} = 0$ whenever $\alpha \neq 0$ or $\beta \neq 0.$ 
Then condition \eqref{eq5.25} becomes $L_0^0 + L_0^{0 \dagger} = 0,$ i.e., 
$L_0^0= - iH$ where $H$ is a bounded selfadjoint operator and equation 
\eqref{eq5.24} takes the form
\begin{equation}
 dU(t) = -iHU(t)  \label{eq5.26}
\end{equation}
which is the familiar Schr\"odinger equation with energy operator $H.$ In view 
of this property we say that \eqref{eq5.24} with condition \eqref{eq5.25} is a 
{\it noisy Schr\"odinger equation.}
\end{remark}

Note that \eqref{eq5.26} has a unique unitary solution with the initial 
condition $U(0) = I$ even when $H$ is any unbounded selfadjoint operator. Thus 
it is natural to examine \eqref{eq5.24} for a unique unitary solution even when 
$L_{\beta}^{\alpha}$ are unbounded operators and \eqref{eq5.25} holds on a 
dense domain or in the form sense. There is some recent progress in this 
direction in the paper \cite{4}.

A case of special interest arises when $L_j^i = 0$ for all $i, j$ in $\{ 1,2, 
\ldots, d\}.$ Writing $L_0^i = L_i$ and
$$\Lambda_0^i = A_i, \quad \Lambda_i^0 = A_i^{\dagger}$$
and using \eqref{eq5.25}, equation \eqref{eq5.24} takes the form
\begin{equation}
dU(t) = \left \{\sum_{j=1}^{n} (L_j d A_j^{\dagger} - L_j^{\dagger} d A_j)-(i H 
+ \frac{1}{2} \sum_{j=1}^{n} L_j^{\dagger} L_j )  dt \right \} U(t)     
\label{eq5.27}
\end{equation}
where $A_j$ and $A_j^{\dagger}$ are the annihilation and creation processes of 
type $j.$ Here the operators $L_j$ and the selfadjoint operator $H$ may all be 
unbounded.

A further specialization to the case when the initial Hilbert space 
$\mathcal{H}_S$ is $\Gamma (\mathbb{C}^n),$ each $L_j$ is of the form $a (u_j) 
+ a^{\dagger} (v_j)$ with $u_j,$ $v_j$ in $\mathbb{C}^n$ and $H$ can be 
expressed as a second degree polynomial of creation and annihilation operators 
in $\Gamma (\mathbb{C}^n)$ is of great interest. F. Fagnola \cite{4} has shown 
the existence of 
unique unitary solutions in the case $n=1$ and has communicated to me that his 
proofs in \cite{4} go through for every finite $n.$ In the next section we 
shall 
show how this can be used to construct quantum Gaussian Markov processes of the 
quasifree type.

We conclude this section with a note on the {\it Heisenberg equation in the 
presence of noise.} Consider the unitary solution $\{U(t)\}$ of \eqref{eq5.24} 
under conditions \eqref{eq5.25}. For any bounded operator $X$ in the system 
Hilbert space $\mathcal{H}_S$ define
\begin{equation}
 j_t (X) = U(t)^{\dagger} (X \otimes I)  U(t), \quad t \ge 0   \label{eq5.28}
\end{equation}
where $I$ is the identity operator in the boson Fock space $\Gamma (L^2 \otimes 
\mathbb{C}^d).$ From conditions \eqref{eq5.25} it follows that there exist 
bounded operators $L_i,$ $1 \leq i \leq d,$ $H$ and $S_j^i,$ $i,j \in \{1,2, 
\ldots, d\}$ such that
\begin{itemize}
 \item[] (i)\quad  $H$ is selfadjoint
 \item[] (ii) \quad the matrix operator $((S_j^i)),$ $i, j \in (1 \leq i \leq 
d)$ in $\mathcal{H}_S \otimes \mathbb{C}^d$ is unitary and
\begin{equation}
  L_j^i = \left \{\begin{array}{lcl} S_j^i - \delta_j^i &\mbox{if} & i, j \in 
\{1,2, \ldots, d\},\\ L_i &\mbox{if}& 1 \leq i \leq d, j=0 \\ - \sum\limits_k 
L_k^{\dagger} S_j^k &\mbox{if}& 1 \leq j \leq d, i = 0, \\ - \{i H + \frac{1}{2}
\sum_{k=1}^{d} L_k^{\dagger} L_k \} &\mbox{if} & i = 0, j=0. \end{array}  
\right .  \label{eq5.29}
\end{equation}      
\end{itemize}
Then it follows  from quantum Ito's formula that
\begin{equation}
dj_t(X) = j_t (\theta_{\beta}^{\alpha} (X)) d \Lambda_{\alpha}^{\beta}    
\label{eq5.30}
\end{equation}
where
\begin{equation}
\theta_{\beta}^{\alpha} (X) = \left \{\begin{array}{l} \sum_{k=1}^{d} 
(S_j^k)^{\dagger} X S_j^k - \delta_j^i X \quad \mbox{if} \quad \alpha=i, 
\beta=j, \\  \sum_{k=1}^{d} (S_i^k)^{\dagger} [X, L_k] \quad \mbox{if} \quad  
\alpha = i, \beta = 0, \\ \sum_{k=1}^{d}    [L_k^{\dagger}, X] S_j^k \quad 
\mbox{if}\quad \alpha = 0, \beta =j, \\ i[H,X] - \frac{1}{2}  \sum_{k=1}^{d}  
(L_k^{\dagger} L_k X + X L_k^{\dagger} L_k - 2 L_k^{\dagger} X L_k) \quad 
\mbox{if}\quad \alpha = \beta = 0 \end{array}   \right .    \label{eq5.31}
\end{equation}
Equations \eqref{eq5.26}, \eqref{eq5.28} and \eqref{eq5.29} describe the 
dynamics of the system observable $X$ in the presence of quantum noise 
generated by the fundamental processes $\Lambda_{\beta}^{\alpha}(t)$ which 
include
time. Suppose $S_j^i=\delta_j^i$ and  $L_{k}=0 \,\forall \,k.$ Then 
\eqref{eq5.30} 
reduces to
$$d j_t (X) = j_t (i [H,X])$$
which is the classical Heisenberg equation with energy operator $H.$ Thus 
\eqref{eq5.28} and \eqref{eq5.29} together deserve to be called a {\it 
Heisenberg equation in the presence of noise} or {\it a Heisenberg-Langevin 
equation.}

Suppose we take the Fock vacuum conditional expectation of $j_t(X)$ and write
$$\langle f |T_t (X) | g \rangle = \langle f \otimes \psi (0) | j_t(X) | g 
\otimes \psi (0)   \rangle. $$
Then $T_t (X)$ is an operator in $\mathcal{H}_S$ and the fact that all the true 
noise operators $\Lambda_{\beta}^{\alpha}(s)$ with $\alpha \neq 0$ or $\beta 
\neq 0$ annihilate the vacuum vector $|\psi (0)\rangle,$ it follows that
$$d T_t (X) = T_t (\theta_0^0 (X)) dt$$
where $\theta_0^0$ is given by the last equation in \eqref{eq5.31}, namely
$$\theta_0^0 (X) = i [H,X] - \frac{1}{2} \sum_{k=1}^{d} (L_k^{\dagger} L_k X + 
XL_k^{\dagger} L_k - 2 L_k^{\dagger} XL_k). $$
In other words $\theta_0^0$ is the generator of the semigroup of completely 
positive maps $\{T_t, t \ge 0 \}.$ The expression $\theta_0^0,$ rather 
remarkably, coincides with the well-known generator in the form obtained by 
Gorini, Kassakowski, Sudarshan \cite{6} and Lindblad \cite{13} in 1976. This 
shows a way to construct stationary quantum Markov processes mediated by 
quantum dynamical semigroups of completely positive maps.

\section{Quantum Gaussian Markov processes and stochastic differential 
equations}
\label{sec:6}
\setcounter{equation}{0}

Consider the boson Fock space $\Gamma (\mathbb{C}^n)$ where $\mathbb{C}^n$ is 
equipped with the canonical orthonormal basis $e_j = (0,0,\ldots, 0,1,0,\ldots, 
0)^T,$ $1$ being the $j$-th coordinate, $j=1,2,\ldots,n.$ Define

\begin{eqnarray*}
 a_j &=& a(e_j), \quad a_j^{\dagger} = a (e_j)^{\dagger} \\
 q_j &=& \frac{a_j + a_j^{\dagger}}{\sqrt{2}}, \quad  p_j = \frac{a_j - 
a_j^{\dagger}}{i\sqrt{2}}
\end{eqnarray*}
Then $p_1, \ldots, p_n$ and $q_1, \ldots, q_n$ satisfy the canonical Heisenberg 
commutation relations and therefore can be viewed as momentum and position 
observables of a quantum system with $n$ degrees of freedom. A state $\rho$ in 
$\Gamma (\mathbb{C}^n)$ is called {\it Gaussian} if its quantum Fourier 
transform 
$\widehat{\rho}(\bm{u}) = \Tr \, \rho W (\bm{u})$ has the form
\begin{equation}
\widehat{\rho}(\bm{u}) = \exp \left \{- i \sqrt{2} \left (\bm{\ell}^T 
\bm{x} 
- \bm{m}^T \bm{y} \right ) - \left ( \bm{x}^T, \bm{y}^T \right ) S {\bm{x} 
\choose \bm{y}}  \right \} \quad \forall \,\,\bm{u} \in \mathbb{C}^n   
\label{eq6.1}
\end{equation}
where $\bm{x} = \Re \bm{u},$ $\bm{y} = \Im \bm{u}$ and $S$ is a real $2n \times 
2n$ matrix satisfying the matrix inequality
\begin{equation}
 2S +i \left [\begin{array}{cc} 0 & -{\rm I_n} \\ {\rm I_n} & 0   \end{array} 
\right ] \ge 
0,  \label{eq6.2}
\end{equation}
$I_n$ being the identity matrix of order $n.$ In such a case one has $\bm{\ell} 
= \Tr \,\bm{p} \rho,$ $\bm{m} = \Tr \, \bm{q} \rho$ and $S$ is the covariance 
matrix of the observables $\left (X_1, X_2, \ldots, X_{2n} \right ) = (p_1, 
p_2, \ldots, p_n, -q_1, \ldots, -q_n)$ so that the $ij$-the entry of $S$ is 
$\Tr \, \frac{X_i X_j + X_j X_i}{2} \rho - \Tr X_i \rho \,\,\Tr X_j \rho.$ So 
we 
call a state $\rho$ satisfying \eqref{eq6.1}, a Gaussian state with {\it 
momentum mean} vector $\bm{\ell},$ {\it position mean} vector $\bm{m}$ and {\it 
covariance matrix} $S.$ For a detailed account of the properties of such 
Gaussian states with a finite degree of freedom we refer to Holevo \cite{8} and 
Parthasarathy \cite{15}, \cite{16}.

Now we look at some transformations of Gaussian states. Define the 
correspondence
\begin{equation}
 T (W(\bm{z})) = W (R^{-1} AR \bm{z} ) \exp \left \{- \frac{1}{2} (R \bm{z})^T 
B (R \bm{z}) \right \}, \bm{z} \in \mathbb{C}^n  \label{eq6.3}
\end{equation}
where $A$ and $B$ are real $2n \times 2n$ matrices, $B \ge 0$
\begin{equation}
R \bm{z} = \left [ \begin{array}{c} \bm{x} \\ \bm{y} \end{array}\right ],\,\,\, 
\bm{x} = \Re \bm{z}, \bm{y} = \Im \,\, \bm{z}.  \label{eq6.4}
\end{equation}
For any mean zero Gaussian state $\rho$ with covariance matrix $S$ we have from 
\eqref{eq6.1}
$$\Tr \, \rho \, T (W(\bm{z})) = \exp  \left \{- (R \bm{z})^T (A^T S A + 
\frac{1}{2} B) R \bm{z} \right \}.$$
The right hand side will be the quantum Fourier transform of a Gaussian state 
with covariance matrix $A^T SA + \frac{1}{2} B$ if and only if
\begin{equation}
2 \left ( A^T SA + \frac{1}{2} B \right ) + i J_{2n} \ge 0    \label{eq6.5}
\end{equation}
where
\begin{equation}
J_{2n} = \left [\begin{array}{cc} 0 & -I_n \\ I_n & 0 \end{array} \right ]   
\label{eq6.6}
\end{equation}
Note that \eqref{eq6.5} holds for every Gaussian covariance matrix $S$ of order 
$2n$ if
\begin{equation}
B+i (J_{2n} - A^T J_{2n} A) \ge 0.   \label{eq6.7}
\end{equation}
It is known from a result of Vanheuerzweijn \cite{19} and the discussion at 
the beginning of the paper \cite{7} by Heinosaari, Holevo and Wolf that for any 
pair $(A,B)$ of $2n \times 2n$ matrices satisfying \eqref{eq6.7}, the 
correspondence defined by \eqref{eq6.3} extends uniquely to a completely 
positive map $T$ on the algebra of all bounded operators on $\Gamma 
(\mathbb{C}^n).$

Now consider a one parameter family $\{T_t\}$ of such completely positive maps 
determined by a family of pairs $(A_t, B_t)$ obeying \eqref{eq6.7}. Simple 
algebra shows that the semigroup condition $T_t T_s = T_{t+s}$ holds if and 
only if the pairs $(A_t, B_t)$ obey the relations
\begin{eqnarray*}
 A_t A_s &=& A_{t+s} \\
B_s + A_s^T B_t A_s &=& B_{t+s} \quad \forall \,\,s,t \ge 0.
\end{eqnarray*}
If we write $A_t = e^{tK}$ then it turns out that any continuous solution $B_t$ 
has the form
$$B_t = \int_0^t e^{sK^{T}} C e^{sK} ds$$
for some matrix $C.$ It is clear that $B_t$ is positive if and only if $C$ is 
positive. Condition \eqref{eq6.7} for the pairs $(A_t, B_t)$ can be written as
$$\int_0^t e^{sK^{T}} \left ( C - i \left \{ K^T J_{2n} + J_{2n} K \right \} 
\right ) e^{sK} ds \ge 0$$
for every $t \ge 0.$ This is equivalent to the condition
$$C +i \left ( K^T J_{2n} + J_{2n} K \right ) \ge 0,$$
since $C,K$ and $J_{2n}$ are all real. Thus we have the following proposition.

\begin{proposition}\label{prop6.1}(Vanheuerzweijn)
Let $(K,C)$ be a pair of $2n \times 2n$ real matrices satisfying the matrix 
inequality
\begin{equation}
 C + i  \left ( K^T J_{2n} + J_{2n} K \right ) \ge 0 \label{eq6.8}
\end{equation}
where $J_{2n}$ is defined by \eqref{eq6.6}. Then there exists a one parameter 
semigroup $\{T_t\}$ of completely positive maps on the algebra of all operators 
in $\Gamma (\mathbb{C}^n)$ satisfying
\begin{equation}
T_t (W(\bm{z})) = W (R^{-1} e^{tK} R \bm{z}) \exp - \frac{1}{2} \left (R 
(\bm{z})^T \int_0^t e^{sK^{T}} C e^{sK} ds \right ) R\bm{z}  \label{eq6.9}
\end{equation}
for all $\bm{z} \in \mathbb{C}^n.$
\end{proposition}
The semigroup $\{T_t\}$ in Proposition \ref{prop6.1} is not strongly 
continuous. However, we can analyse the time derivatives of expressions of the 
form $\langle e (\bm{u}) | T_t (W(\bm{z})) | e(\bm{v})\rangle.$ The analysis of 
such 
derivatives leads to a formal quantum stochastic differential equation of the 
noisy Heisenberg type.

\begin{proposition}\label{prop6.2}
Let $\{T_t\}$ be the semigroup defined by \eqref{eq6.9}. Then 
\begin{eqnarray*}
\lefteqn{\frac{d}{dt} \langle e (\bm{u}) \left | T_t (W(\bm{z})) | e (\bm{v}) 
\rangle \right |_{t=0}} \\
&=& \langle  e (\bm{u}) \left | \mathcal{L} (W(\bm{z})) \right | e(\bm{v})    
\rangle \quad \forall \,\,\bm{z}, \bm{u}, \bm{v} \in \mathbb{C}^n,
\end{eqnarray*}
where 
\begin{eqnarray*}
 \mathcal{L} (W(\bm{z})) &=& \left \{a^{\dagger} (R^{-1} K R\bm{z} ) - a 
(R^{-1} K R \bm{z}) \right .\\
&& \left .+ \frac{1}{2} \left \{ \langle R^{-1} KR \bm{z}|z \rangle \right .   
- 
\langle \bm{z} | R^{-1} KR \bm{z} \rangle - (R\bm{z})^T C R \bm{z}  \right \} W 
(\bm{z}).
\end{eqnarray*}
\end{proposition}

\begin{proof}
 This is straightforward differentiation using the definitions of the operators 
$a(\cdot),$ $a^{\dagger}(\cdot),$ $W(\cdot)$ and the commutation relation $W 
(\bm{u}) a(\bm{v}) W(\bm{u})^{-1} = a (\bm{v})- \langle \bm{v} | \bm{u} 
\rangle$ for all $\bm{u}, \bm{v}$ in $\mathbb{C}^n.$
\end{proof}

\begin{proposition}\label{prop6.3}
Let $\{T_t\}$ be the completely positive semigroup determined by the pair 
$(K,C)$ in Proposition \ref{prop6.1} by \eqref{eq6.9}. For any state $\rho$ in 
$\Gamma (\mathbb{C}^n)$ let $T_t^{\prime} (\rho)$ be the state defined by
$$\Tr \,\,T_t^{\prime} (\rho) W(\bm{z}) = \Tr \, \rho \, T_t (W(\bm{z})) \quad 
\forall \,\,\bm{z} \in \mathbb{C}^n.$$
Denote by $\rho_g (\bm{\ell}, \bm{m}, S)$ the Gaussian state with momentum mean 
vector $\bm{\ell},$ position mean vector $\bm{m}$ and covariance matrix $S.$ 
Then
$$T_t^{\prime} \left ( \rho_g (\bm{\ell}, \bm{m}, S) \right )  = \rho_g 
(\bm{\ell}_t, \bm{m}_t, S_t)$$
where
\begin{eqnarray*}
 \left [ \begin{array}{c} \bm{\ell}_{t} \\ -\bm{m}_{t} \end{array} \right ] &=&
e^{tK^{T}} \left [\begin{array}{c} \bm{\ell} \\ - \bm{m}\end{array} \right ], \\
S_t &=& e^{tK^T} S e^{tK} + \frac{1}{2} \int_0^t e^{sK^{T}} C e^{sK} ds, \,\, 
t 
\ge 0.
\end{eqnarray*}
\end{proposition}

\begin{remark}
 It is known from a general theory \cite{3} that $\{T_t\}$ can be dilated to a 
quantum Markov process. If the initial state of the Markov process is Gaussian 
it follows that the state of the system at any time $t$ is Gaussian. We may 
call the dilation of $\{T_t\}$ a quantum Gaussian Markov process.

We shall now present examples when the dilated process  can be obtained by a 
unitary evolution driven by a quantum stochastic differential equation.
\end{remark}

Consider the Hilbert space
$$\widehat{\mathcal{H}} = \Gamma (\mathbb{C}^n) \otimes \Gamma (L^2 
(\mathbb{R}_{+})).$$
Let 
\begin{equation}
L = a (\bm{u}) + a^{\dagger} (\bm{v}),   \label{eq6.10}
\end{equation}
where  $\bm{u}, \bm{v}$ are fixed in $\mathbb{C}^n.$ Then $L$ is an operator in 
the initial Fock space $\Gamma (\mathbb{C}^n).$ Let
\begin{eqnarray*}
 A(t) &=& a (1_{[0,t]}), \\
A^{\dagger}(t) &=& a^{\dagger} (1_{[0,t]}), \quad  t \ge 0
\end{eqnarray*}
be the annihilation and creation processes in the Fock space $\Gamma (L^2 
(\mathbb{R}_{+})).$ Following the notations of Section 5 consider the 
specialized  form of the stochastic differential equation \eqref{eq5.27}, 
Section 5:
\begin{equation}
d U(t) = (L d A^{\dagger} - L^{\dagger} d A - \frac{1}{2} L^{\dagger} L dt ) U 
(t), t \ge 0   \label{eq6.11}
\end{equation}
with $U(0)=I$ where $L$ is now the unbounded operator given by \eqref{eq6.10}. 
Then, by Fagnola's theorem \cite{4} it follows that there exists a unique 
adapted unitary operator-valued process satisfying \eqref{eq6.11} and therefore 
one can define the operators
$$j_t (X) = U(t)^{\dagger} (X \otimes 1) U(t), t \ge 0$$
for any bounded operator $X$ in $\Gamma (\mathbb{C}^n).$ If we choose 
$X=W(\bm{z}), \bm{z} \in \mathbb{C}^n,$ then the Markovian dynamical semigroup 
$\{T_t^L, t \ge 0\}$ determined by the flow $\{j_t\}$ satisfies the Lindblad 
equation
\begin{eqnarray*}
 \frac{d}{dt} T_t^L  \left .(W(\bm{z})) \right |_{t=0} &=& - \frac{1}{2} \left 
\{L^{\dagger} L W(\bm{z}) + W(\bm{z}) L^{\dagger} L - 2 L^{\dagger} W(\bm{z}) 
L\right \} \\
&=& - \frac{1}{2} \left \{L^{\dagger} [L, W(\bm{z})] + [W(\bm{z}), L^{\dagger}] 
L  \right \}.
\end{eqnarray*}
The commutation relations between $a(\bm{u})$ and $W(\bm{z})$ lead to the 
relations:
\begin{eqnarray}
\lefteqn{ \frac{d}{dt} \left .\left (T_t^L( W(\bm{z}) \right ) \right |_{t=0}   
      } \nonumber \\  
&=& \left \{ a^{\dagger} \left (\frac{\overline{\lambda (\bm{z}) }\bm{v} - 
\lambda (\bm{z})\bm{u}}{2} \right ) - a  \left (\frac{\overline{\lambda 
(\bm{z}) }\bm{v} - \lambda 
(\bm{z})u}{2}  \right ) - \frac{1}{2} |\lambda (\bm{z})|^2 
\right \} W (\bm{z})            \label{eq6.12}
\end{eqnarray}
for all $\bm{z} \in \mathbb{C}^n,$ where

\begin{equation}
\lambda(\bm{z}) = \langle \bm{u}|\bm{z}  \rangle + \langle \bm{z} | 
 \bm{v}
\rangle   \label{eq6.13}
\end{equation}
and the equations are to be understood in the weak sense on the exponential 
domain $\mathcal{E},$ the linear manifold generated by exponential vectors in 
$\Gamma (\mathbb{C}^n).$ Comparing \eqref{eq6.12} and the equation for the 
generator $\mathcal{L}$ in Proposition \ref{prop6.2} we see that $\mathcal{L} 
(W(\bm{z}))$ coincides with the right hand side of \eqref{eq6.12} if the 
following hold for all $\bm{z} \in \mathbb{C}^n$:

\begin{eqnarray}
R^{-1} KR \bm{z} &=& \frac{\overline{\lambda (\bm{z})} \bm{v} - \lambda 
(\bm{z}) \bm{u}}{2},   \label{eq6.14} \\
\frac{1}{2} \left \{ \langle R^{-1} K R \bm{z} | \bm{z}  \rangle \right . &-& 
\left . \langle    \bm{z}  | R^{-1} KR  \bm{z}  \rangle - (R \bm{z}  )^T C R  
\bm{z} \right \} = - \frac{1}{2} \left | \lambda(\bm{z})\right |^2  
\label{eq6.15}
\end{eqnarray}
where $\lambda(\bm{z})$ is given by \eqref{eq6.13} and $(K,C)$ is the pair of 
$2n \times 2n$ real matrices in Proposition \ref{prop6.1}. Solving for the pair 
$(K,C)$ in \eqref{eq6.14} and \eqref{eq6.15} for a given pair $\bm{u}, \bm{v}$ 
in $\mathbb{C}^n$ 
one obtains the following

\begin{eqnarray}
K &=&  \left [\begin{array}{cc}  - \Re 
(\overline{\bm{u}}-\bm{v})(\bm{u} + \overline{\bm{v}})^T & -\Im 
(\overline{\bm{u}}-\bm{v})(\bm{u} - \overline{\bm{v}})^T 
\\ \Im  (\overline{\bm{u}}+\bm{v})(\bm{u} + \overline{\bm{v}})^T & -\Re 
(\overline{\bm{u}}+\bm{v})(\bm{u} -
\overline{\bm{v}})^T \end{array} \right ]   \label{eq6.16} \\
C &=& \left [\begin{array}{cc}\Re (\bm{u} +\overline{\bm{v}}) ( 
\overline{\bm{u}} + \bm{v})^T & -\Im  
(\bm{u}+\overline{\bm{v}})(\overline{\bm{u}}-\bm{v})^T \\
\Im (\bm{u}- \overline{\bm{v}})(\overline{\bm{u}} + \bm{v})^T & \Re 
(\bm{u}-\overline{\bm{v}} )(\overline{\bm{u}} - \bm{v})^T
\end{array} \right ],   
\label{eq6.17}
\end{eqnarray}
and
\begin{equation}
C + i \left ( K^T J_{2n} + J_{2n} K \right ) =  \left [\begin{array}{c} 
\bm{u}+\overline{\bm{v}} \\ -i (\bm{u}-\overline{\bm{v}})  
\end{array} \right ] \left [\begin{array}{c} 
\bm{u}+\overline{\bm{v}} \\ -i (\bm{u}-\overline{\bm{v}})  
\end{array} \right ] ^{\dagger} \ge 0. \label{eq6.18}
\end{equation}
We denote by $\left (K (\bm{u}, \bm{v}), C (\bm{u}, \bm{v}) \right )$ the pair 
$(K,C)$ determined by \eqref{eq6.16} and \eqref{eq6.17}. Thus the Markov 
process determined by the pair $(K(u,v), C (u,v))$  according to the semigroup 
$\{T_t\}$ in  \eqref{eq6.9}  of 
Proposition \ref{prop6.1} is completely described by the noisy Schr\"odinger 
equation
\begin{eqnarray}
dU(t) &=& \left \{ (a (\bm{u}) + a^{\dagger} (\bm{v}) ) dA^{\dagger} - 
(a^{\dagger} (\bm{u}) + a ( \bm{v})) dA \right . \nonumber \\
&&  \left . -\frac{1}{2} (a^{\dagger} (\bm{u}) + a (\bm{v})) (a (\bm{u}) + 
\bm{a}^{\dagger} ( \bm{v})) dt \right \} U(t).  \label{eq6.19}
\end{eqnarray}
with $U(0)=I.$

We shall now consider the general case of a pair $(K,C)$ of $2n \times 2n$ real 
matrices satisfying the condition \eqref{eq6.8} and show how the semigroup 
$\{T_t\}$ in \eqref{eq6.9} can be dilated to a quantum Markov process driven 
by a noisy Schr\"odinger equation. To this end we write
$$C + i (K^T J_{2n}+J_{2n} K) = \sum_{j=1}^r \left [\begin{array}{c} \bm{b}_j 
\\ 
\bm{c}_j \end{array} \right ] \left [\begin{array}{c} \bm{b}_j \\ 
\bm{c}_j \end{array} \right ]^{\dagger}$$
where $\left \{ \left [\begin{array}{c} \bm{b}_j \\ 
\bm{c}_j \end{array} \right ], 1 \leq j \leq r \right \}$ is a set of mutually 
orthogonal vectors in $\mathbb{C}^{2n}$ so that $\bm{b}_j, \bm{c}_j$ belong to 
$\mathbb{C}^n$ and $r$ is the rank of the left hand side when it is not the 
zero matrix. If the left hand side vanishes then $C=0$ and $K$ is an element of 
the Lie algebra $sp(2n).$ Now write

\begin{equation}
L_j = a (\bm{b}_j) + a^{\dagger} ( \bm{c}_j), 1 \leq j \leq r \quad 
\mbox{if}\,\, r \neq 0.  \label{eq6.20}
\end{equation}
and define the matrices $K(\bm{u}_j, \bm{v}_j),$ $C(\bm{u}_j, \bm{v}_j)$ as the 
right hand side of \eqref{eq6.16} and \eqref{eq6.17} respectively with $\bm{u}, 
\bm{v}$ replaced by $\bm{u}_j,\bm{v}_j$ respectively and  $\bm{u}_j, \bm{v}_j$ 
are determined by the equations
\begin{eqnarray*}
 \bm{u}_j + \overline{\bm{v}_j} &=& \bm{b}_j, \\
\bm{u}_j - \overline{\bm{v}_j} &=& i\bm{c}_j
\end{eqnarray*}

Then
\begin{eqnarray}
 C &=& \sum_{j=1}^r C (\bm{u}_j,\bm{v}_j)  \label{eq6.21} \\ 
 K &=& \sum_{j=1} ^{r} K (\bm{u}_j,\bm{v}_j) + K^{\prime}    \label{eq6.22}
\end{eqnarray}
where $K^{\prime} \in sp (2n).$ Then $\{e^{tK^{\prime}}, t \in \mathbb{R}\}$ is 
a one parameter group in $Sp(2n).$ 

Now we observe that equation (\ref{eq6.16}) implies that the
matrix $J_{2n} K(\bm{u},\bm{v})$ is skew symmetric and hence
it follows that $J_{2n}\sum_{j=1}^r K(\bm{u}_j, \bm{v}_j)$
is skew symmetric. Write 
\begin{eqnarray}
M &=& J_{2n} \sum_{j=1}^r  K(\bm{u}_j, \bm{v}_j),
\label{eq6.23}\\
N &=& J_{2n}K'. \label{eq6.24}
\end{eqnarray}
Then $M$ is skew symmetric, $N$ is symmetric and $J_{2n}K=M+N$.
Thus $M$ and $N$ are given by
\begin{eqnarray}
M &=& \frac{J_{2n}K - (J_{2n}K)^T}{2}, \label{eq6.25}\\
N &=& \frac{J_{2n}K + (J_{2n}K)^T}{2}. \label{eq6.26}
\end{eqnarray}
Consider the spectral decomposition of $N$
\begin{equation}
N=\sum_{j=1}^{r'} \lambda_j 
\begin{bmatrix} 
\bm{\beta}_j \\ \bm{\gamma}_j 
\end{bmatrix} 
\begin{bmatrix} 
\bm{\beta}_j \\ \bm{\gamma}_j 
\end{bmatrix}^T, \qquad \bm{\beta}_j,~ \bm{\gamma}_j \in \mathbb{R}^n \label{eq6.27}
\end{equation}
where $\lambda_j$'s are the nonzero eigenvalues of $N$ and $\begin{bmatrix} 
\bm{\beta}_j \\ \bm{\gamma}_j 
\end{bmatrix}$, the respective eigenvectors. Put
\begin{equation}
\bm{w}_j = \bm{\beta}_j + i \bm{\gamma}_j, \qquad 1 \leq j \leq r'.
\label{eq6.28}
\end{equation}
Define the self-adjoint operator $H$ by   
\begin{equation}
H=\frac{1}{4} \sum_{j=1}^{r'} \lambda_j (a(\bm{w}_j)
+a^\dag (\bm{w}_j))^2. \label{eq6.29} 
\end{equation}
A simple (but tedious) algebra using CCR shows that 
\begin{multline}
-i[H,W(\bm{z})] =\left\{ a^\dag(R^{-1}K'R\bm{z})
-a(R^{-1}K'R\bm{z}) \right. \\ \left. +\frac{1}{2} (\langle
R^{-1}K'R\bm{z}| \bm{z}
\rangle - \langle \bm{z} | R^{-1}K'R\bm{z} \rangle ) \right\}
W(\bm{z}), \label{eq6.30}
\end{multline}
where $K'$ is given by Equation (\ref{eq6.24}). 

\par Now consider
the noisy Schr\"odinger equation
\begin{equation}
dU(t)=\left\{ \sum_{j=1}^r \left(L_jdA_j^\dag -L_j^\dag
dA_j\right) - \left( iH + \frac{1}{2} \sum_{j=1}^r L_j^\dag
L_j \right)dt \right\}
U(t),
\label{eq6.31}
\end{equation}
where $L_j$ is given by Equation (\ref{eq6.20}), $H$ is
defined by (\ref{eq6.26} -- \ref{eq6.29}), and $A_j,~ A_j^\dag$
are the annihilation and creation noise processes in the
boson Fock space
\[\hat{\mathcal{H}} = \Gamma(\mathbb{C}^n) \otimes
\Gamma(L^2(\mathbb{R}_+) \otimes \mathbb{C}^r),\]
$r$ being the rank of the matrix $C + i (K^T J + JK)$. When
$r=0,~ \hat{\mathcal{H}} =\Gamma(\mathbb{C}^n)$ and
(\ref{eq6.31}) reduces to the Schr\"odinger's equation 
\[ dU(t) = -i H U(t) dt,\]
$C=0$, $K$ ( $=K'$) is symplectic and $H$ is determined by
(\ref{eq6.29}). (In this special case the formula for $H$
seems to be more precise than what is presented in
\cite{2}.)
 It now follows from Fagnola's theorem \cite{4} for
the case of finite number of degrees of freedom that there
exists a unitary solution $\{ U(t), t\geq 0\}$ for the
equation (\ref{eq6.31}) and 



$$T_t (X) = \mathbb{E}_{0]} U(t)^{\dagger} (X \otimes 1) \, U(t), \quad t \ge 
0$$
where $\mathbb{E}_{0]}$ is the Fock vacuum conditional expectation of 
quantum stochastic calculus. Thus we have described completely the dilation of 
the completely positive semigroup in Proposition \ref{prop6.1}. In view of 
Proposition \ref{prop6.3}, such a dilation is a quantum Gaussian Markov process.

\end{document}